\documentclass[conference]{IEEEtran}
\IEEEoverridecommandlockouts
\usepackage{cite}
\usepackage{amsmath,amssymb,amsfonts}
\usepackage{algorithmic}
\usepackage{graphicx}
\usepackage{textcomp}
\usepackage{xcolor}

\usepackage{amsthm}
\usepackage{hyperref}
\usepackage{multirow}
\usepackage{subfigure}
\usepackage{footnote}

\def\BibTeX{{\rm B\kern-.05em{\sc i\kern-.025em b}\kern-.08em
    T\kern-.1667em\lower.7ex\hbox{E}\kern-.125emX}}

\DeclareRobustCommand*{\IEEEauthorrefmark}[1]{%
    \raisebox{0pt}[0pt][0pt]{\textsuperscript{\footnotesize\ensuremath{#1}}}}

\begin{document}
\newtheorem{theorem}{Theorem}
\theoremstyle{definition}
\newtheorem{definition}{Definition}
\title{Privacy-Preserving Cross-Domain Sequential Recommendation\\
}

\author{
	\IEEEauthorblockN{
	Zhaohao Lin\IEEEauthorrefmark{1}, 
	Weike Pan\IEEEauthorrefmark{1,*} and
	Zhong Ming\IEEEauthorrefmark{1,2,*}\thanks{*: Co-corresponding authors}}
	\IEEEauthorblockA{\IEEEauthorrefmark{1}\textit{College of Computer Science and Software Engineering, Shenzhen University}}
	\IEEEauthorblockA{\IEEEauthorrefmark{2}\textit{Guangdong Laboratory of Artificial Intelligence and Digital Economy (SZ), Shenzhen University}}
	\IEEEauthorblockA{Shenzhen, China}
	\IEEEauthorblockA{linzhaohao2021@email.szu.edu.cn, \{panweike,mingz\}@szu.edu.cn}
}

\maketitle

\begin{abstract}
	Cross-domain sequential recommendation is an important development direction of recommender systems.
	It combines the characteristics of sequential recommender systems and cross-domain recommender systems, which can capture the dynamic preferences of users and alleviate the problem of cold-start users.
	However, in recent years, people pay more and more attention to their privacy.
	They do not want other people to know what they just bought, what videos they just watched, and where they just came from.
	How to protect the users' privacy has become an urgent problem to be solved.
	In this paper, we propose a novel privacy-preserving cross-domain sequential recommender system (PriCDSR), which can provide users with recommendation services while preserving their privacy at the same time.
	Specifically, we define a new differential privacy on the data, taking into account both the ID information and the order information.
	Then, we design a random mechanism that satisfies this differential privacy and provide its theoretical proof.
	Our PriCDSR is a non-invasive method that can adopt any cross-domain sequential recommender system as a base model without any modification to it.
	To the best of our knowledge, our PriCDSR is the first work to investigate privacy issues in cross-domain sequential recommender systems.
	We conduct experiments on three domains, and the results demonstrate that our PriCDSR, despite introducing noise, still outperforms recommender systems that only use data from a single domain.
\end{abstract}


\section{Introduction}
\label{sec:intro}

Today, recommender systems have emerged as a crucial technology that can offer significant practical benefits to commercial enterprises.
It has played an important role in many applications, including e-commerce, online video platforms, social media, and mapping services.
A recommender system typically employs a model that leverages user-item interaction records, user profiles, item attributes, and other relevant information to predict user preferences and then provide personalized recommendation services for users.

Considering that a user's preferences may change over time, an increasing number of recent works pay more attention to temporal information and propose many models that exploit sequential information.
A recommender system using sequential information is called a sequential recommender system.
By combining the users' long-term and short-term preferences, sequential recommender systems can more effectively capture a user's interests and thus provides better recommendation service.
However, sequential recommender systems still suffer from the cold-start and data sparsity problems.
Fortunately, there is another research direction in recommender systems that can effectively mitigate these two problems, namely, cross-domain recommender systems (CDR).
By leveraging additional auxiliary data from other business applications, CDR can enhance the model performance.
These applications are called auxiliary domains, while the application to improve recommendation accuracy is called target domain.
In recent years, some researchers have successfully combined sequential recommender systems and cross-domain recommender systems, leading to the emergence of cross-domain sequential recommender systems (CDSR).
The models of CDSR not only effectively utilize the users' long-term and short-term preferences, but also transfer user and item knowledge from other domains, resulting in improved recommendation performance.

Existing cross-domain sequential recommender systems usually assume that the data of each domain is visible to all the domains.
In other words, the data is public.
However, this practice may lead to legal problems in real-world scenarios, because the data usually includes the users' historical interacted items, which are the privacy of the users.
After the implementation of data protection laws and regulations such as the General Data Protection
Regulation (GDPR) \footnote{\url{https://gdpr-info.eu}}, it is illegal to collect, process or exchange a user's data without the user's consent.
As far as we know, no researchers have proposed studies or given ideas to tackle this new challenge.
Therefore, we aim to propose a viable solution to this new challenge, i.e., privacy-preserving cross-domain sequential recommender systems.

In this paper, we propose a privacy-preserving cross-domain sequential recommender system (PriCDSR), which can be applied to CDSR to provide users with recommendation services while preserving their privacy.
To achieve this goal, we introduce the definition of differential privacy in users' historical interaction sequences, taking into account both the ID and the order information.
Then, we design a random mechanism that satisfies this differential privacy and provide its theoretical proof.

Our main contributions can be summarized as follows:
(i) to the best of our knowledge, our PriCDSR is the first work on privacy-preserving cross-domain sequential recommender systems;
(ii) we provide the theoretical proof that our PriCDSR satisfies the differential privacy in users' historical interaction sequences;
(iii) we conduct experiments on six domain couples and the results demonstrate that our PriCDSR, despite introducing noise, still outperforms recommender systems that only use data from a single domain.

\section{Related Work}
\label{sec:related}

\subsection{Cross-Domain Sequential Recommendation}

\subsubsection{Merged Flow Cross-Domain Sequential Recommendation}
\label{sec:mfcdsr}
There exist some works that merge the data in the target domain and the auxiliary data, which we call merged flow cross-domain sequential recommendation (MFCDSR).
Specifically, the interaction records of each user in both domains are rearranged in chronological order by timestamps.
$\pi$-net~\cite{piNet} considers the shared account problem in smart TV recommendation.
The authors propose shared account filter unit (SFU) to capture different user preferences within an account, and cross-domain transfer unit (CTU) to model the transfer of user preferences between different domains.
PSJNet~\cite{PSJNet} is a framework for MFCDSR.
The authors reformulate $\pi$-net, regard SFU and CTU as a split-by-join unit, and then propose PSJNet-II ($\pi$-net is regarded as PSJNet-I) that replaces the split-by-join unit with a split-and-join scheme.
In addition to the users' historical interaction sequences, MIFN~\cite{MIFN} also incorporates the item information represented as a knowledge graph.
DA-GCN~\cite{DA-GCN} and DDGHM~\cite{DDGHM} are graph-based models, and they use graphs for better user and item representations.
C$^2$DSR~\cite{C2DSR} utilizes graph neural network and self-attention mechanism to obtain single-domain and cross-domain representations of users.

\subsubsection{Segregated Flow Cross-Domain Sequential Recommendation}
\label{sec:sfcdsr}

If a model does not utilize the timestamps to merge data from the two domains, we refer to it as a segregated flow cross-domain sequential recommendation (SFCDSR).
MiNet~\cite{MiNet} models three types of user interest (i.e., long-term interest across domains, short-term interest from the auxiliary domain, and short-term interest in the target domain), and then aggregates them to obtain the final user interest.
SEMI~\cite{SEMI} is the first work that investigates e-commerce micro-video recommendation.
It enhances the similarity between the representations of a user across different domains in the same session, while reducing the similarity to the representations of other users or sessions.
In DASL~\cite{DASL}, the authors propose two techniques called dual embedding (DE) and dual attention (DA).
Specifically, DE uses an orthogonal matrix to map the user embeddings in one domain to the other domain, and then minimizes the gap between the mapped embeddings and the embeddings in the other domain.
CD-ASR~\cite{CD-ASR} and CD-SASRec~\cite{CD-SASRec} employ self-attention mechanism for domain-specific user representations in both auxiliary and target domains.
The two representations are then combined for improved recommendation.
RecGURU~\cite{RecGURU} introduces adversarial learning into cross-domain sequential recommendation.
It uses an encoder to obtain generalized user representations (GUR) from the users' historical interaction sequences, and a discriminator to identify which domain the GURs come from.
Then, RecGURU alternately optimizes the model and the discriminator to achieve knowledge transfer.

The above works are all classic or recent works on cross-domain sequential recommender systems.
However, both MFCDSR and SFCDSR do not consider privacy issues, so it is difficult to legally use them in real-world scenarios.

\subsection{Privacy-Preserving Cross-Domain Recommendation}
\label{sec:ppcdr}

While there is a lack of research on privacy-preserving cross-domain sequential recommender systems, several studies have explored privacy concerns in cross-domain recommender systems (CDR).

PriCDR~\cite{PriCDR} is a two-stage model.
In the first stage, PriCDR uses Johnson-Lindenstrauss transform (JLT) and sparse-aware JLT (SJLT) to add noise to the rating matrix of the auxiliary domain to satisfy differential privacy.
In the second stage, PriCDR uses a proposed model called HeteroCDR to model the user preferences using the noise-added rating matrix of the auxiliary domain and the original rating matrix of the target domain.
PPGenCDR~\cite{PPGenCDR} is also a two-stage model.
In the first stage, it uses SPPG, a module based on generative adversarial network (GAN), to extract the user preferences in the auxiliary domain (which are represented as the generator of SPPG), and then transfers the extracted user preferences (i.e., the generator) to the target domain.
In the second stage, a robust CDR module in the target domain utilizes the received generator and the raw data of the target domain to enhance the recommendation performance of the target domain.
P2FCDR~\cite{P2FCDR} is the first work to study the privacy-preserving problem of dual-target cross-domain recommender systems.
It uses local differential privacy (LDP) to protect user embeddings exchanged between the two domains.
The approach also uses a feature-level gated selecting vector to refine the information fusion process of the transferred embeddings.
The work~\cite{eggs} proposes a method for reconstructing user composite embeddings for cross-domain recommender systems.
The authors propose a deep learning network for reconstructing user embeddings from various auxiliary domains into a composite user embedding, thereby leveraging information from multiple auxiliary domains to improve the accuracy of downstream tasks.

There are two other works, FedCDR~\cite{FedCDR} and FedCT~\cite{FedCT}, that also consider privacy issues, but they adopt a different research problem setting than ours.
In their research problem, each client is free to use data from both domains.
However, our research problem does not consider the clients.

\subsection{Differential Privacy for Sequential Data}
\label{sec:dps}

Several studies have focused on preserving the privacy of point-of-interest (POI) data~\cite{poi1,poi2,poi3}.
However, most of these studies aim to protect the geographic coordinates (i.e., latitudes and longitudes) of the POIs.
Our research, on the other hand, is framed in a more general setting that is not restricted to any particular application scenario.
As a result, these methods are not effective in addressing our research problem since the items in CDSR do not have latitude and longitude information.

In the work~\cite{CLDP}, the authors propose the notion of condensed local differential privacy (CLDP) and propose a mechanism to satisfy CLDP.
The mechanism is designed to protect two types of sensitive information of sequential data, i.e., length and content.
However, it does not protect the order of items in the sequences, which may lead to privacy issues.

In another work~\cite{TLDP}, the authors propose the notion of the local differential privacy in the temporal setting (TLDP) and propose three mechanisms to satisfy TLDP.
However, TLDP only considers the privacy of the order of the items, but not the privacy of the IDs of the items.
Specifically, TLDP allows the replacement of an item in a sequence only with another item from the same sequence, but not with an item from the entire item set that is not present in the sequence.

\section{Preliminaries}
\label{sec:preliminaries}
\subsection{Problem Definition}
\label{sec:proDef}
In this paper, we consider a cross-domain sequential recommendation problem.
Without loss of generality, we assume that there are only two domains, an auxiliary domain and a target domain.
These two domains have the same set of users, denoted as $\mathcal{U}$, where $|\mathcal{U}|=n$.
However, the item sets of them are mutually exclusive.
The item set of the auxiliary domain is denoted as $\mathcal{I}^A$, where $|\mathcal{I}^A|=m_A$.
And that of the target domain is denoted as $\mathcal{I}^T$, where $|\mathcal{I}^T|=m_T$.
In each domain, we only consider a kind of user-item one-class feedback.

Our goal is to improve the recommendation performance on the target domain using the data on the auxiliary domain, while also protecting the data privacy of the auxiliary domain.
We believe that there are no privacy issues within a domain.
Therefore, all data within a domain is freely available for use within that domain.
Note that although we only consider the case of one target domain and one auxiliary domain, it can be easily extended to the case of one target domain and multiple auxiliary domains.

\subsection{Sequential Differential Privacy}

The items that user $u$ interacts with in the auxiliary domain (denoted as $\mathcal{I}^A_u$), is obviously the privacy of user $u$ and thus needs to be protected.
We call this kind of information ID information.
In addition, we must consider the sequential information of users' historical interaction sequences, which reveals their interests.
For example, in POI recommendation, a user's interaction subsequence from $A$ to $B$ indicates that the user has traveled from $A$ to $B$.
On the contrary, a subsequence from $B$ to $A$, despite the same POIs being involved, denotes that the user has moved from $B$ to $A$.
Obviously, the information of $A \rightarrow B$ and $B \rightarrow A$ subsequences is quite different.
This scenario is not exclusive to POI recommendation and may arise in other types of recommendation tasks as well.
This kind of users' interests in the auxiliary domain should also be considered private, so we also need to protect the sequential information.

Using differential privacy to only protect the ID information without regarding to the order information may leak the user's privacy.
It seems that perturbing the item IDs in the users' historical interaction sequences automatically perturbs the order of the item IDs.
However, this is only a side effect of protecting the ID information, and does not adequately protect the order information, let alone guarantee it theoretically.
To this end, we propose sequential differential privacy (SDP), a new differential privacy that protects both the ID and sequential information of the users' interaction sequences.
We have a sequential data matrix as defined in Definition \ref{def:sdm}.
Under such setting, we also define the neighbouring sequential data matrices in Definition \ref{def:nsdm}.
Finally, we define SDP in Definition \ref{def:sdp}.

\begin{definition}[Sequential Data Matrix]
	\label{def:sdm}
	The data in the auxiliary domain is denoted as a set $\mathcal{R}^A = \{\mathcal{R}^A_u|u \in \{1,2,\ldots,n\}\}$.
	Each user $u$ has an ordered sequence of interacted items in the auxiliary domain, i.e., $\mathcal{R}^A_u = (i^1_u,i^2_u,\ldots,i^{|\mathcal{R}^A_u|}_u)$.
	We fix the length of the users' historical interaction sequences as $L$.
	We only take the last $L$ item IDs for the interaction sequence if that is too long, and pad zeros in the front of the sequence if it is too short.
	We only keep the first interaction record between a user and an item, and delete the subsequent repeated interaction records, which is a common setting in sequential recommender systems~\cite{SASRec,FISSA}.
	Therefore, the data in the auxiliary domain can be denoted as a sequential data matrix $\mathbf{R}^A\in\mathbb{N}^{n \times L}$, where $\mathbf{R}^A_{u,\ell}=i^\ell_u$.
\end{definition}

\begin{definition}[Neighbouring Sequential Data Matrices]
	\label{def:nsdm}
	Two sequential data matrices, $\mathbf{R}$ and $\mathbf{R}'$, are neighbouring if $\mathbf{R}'$ can be obtained from $\mathbf{R}$ by modifying one user-item interaction or by swapping any two interactions of a single user.
	Formally, in the former situation, there exists one pair $(u_0,\ell_0)$ with $1 \leq u_0 \leq n$ and $1 \leq \ell_0 \leq L$ such that
	$$
		\begin{cases}
			\mathbf{R}_{u,\ell} \neq \mathbf{R}'_{u,\ell}, & u=u_0 \land \ell = \ell_0        \\
			\mathbf{R}_{u,\ell} =  \mathbf{R}'_{u,\ell},   & u \neq u_0 \lor \ell \neq \ell_0
		\end{cases}
	$$
	In the later situation, there exists two pairs $(u_1, \ell_1)$ and $(u_1, \ell_2)$ with $1 \leq u_1 \leq n$ and $1 \leq \ell_1 < \ell_2 \leq L$ such that
	$$
		\begin{cases}
			\mathbf{R}_{u,\ell} = \mathbf{R}'_{u,\ell'}, & u=u_1 \land \ell = \ell_1 \land \ell' = \ell_2            \\
			\mathbf{R}_{u,\ell} = \mathbf{R}'_{u,\ell'}, & u=u_1 \land \ell = \ell_2 \land \ell' = \ell_1            \\
			\mathbf{R}_{u,\ell} =  \mathbf{R}'_{u,\ell}, & u \neq u_1 \lor (\ell \neq \ell_1 \land \ell \neq \ell_2)
		\end{cases}
	$$
\end{definition}

\begin{definition}[Sequential Differential Privacy, SDP]
	\label{def:sdp}
	For any pair of neighbouring sequential data matrices, $\mathbf{R}$ and $\mathbf{R}'$, and any output $\mathbf{A}$, a random mechanism $\mathcal{M}$ can provide $\epsilon$-SDP (i.e., the random mechanism $\mathcal{M}$ is an $\epsilon$-SDP random mechanism) if we have:
	$$
		\forall \mathbf{A}\in Range(\mathcal{M}(\cdot)), \frac{\Pr[\mathcal{M}(\mathbf{R})=\mathbf{A}]}{\Pr[\mathcal{M}(\mathbf{R}')=\mathbf{A}]} \leq \exp(\epsilon)
	$$
\end{definition}

The two types of changes in the defined neighbouring sequential data matrices have significant implications for SDP, as they respectively highlight the random mechanism's attention to ID and sequential information.
In particular, the former change demonstrates that SDP is sensitive to the values of item IDs, while the latter shows that SDP also pays attention to the sequential information of the users' historical interaction sequences.

\section{PriCDSR}
\label{sec:method}
In this section, we describe the steps of our PriCDSR in detail.
Under the definition of sequential data matrix in Definition \ref{def:sdm}, there is a sequential data matrix $\mathbf{R}^T \in \mathbb{N}^{n \times L}$ in the target domain, and a sequential data matrix $\mathbf{R}^A \in \mathbb{N}^{n \times L}$ in the auxiliary domain.

The overall architecture of our PriCDSR is shown in Fig. \ref{fig:pricdsr}.
Inspired by PriCDR~\cite{PriCDR}, our PriCDSR also adopts a two-stage solution.
In the first stage, our PriCDSR adds a certain level of noise to the auxiliary domain's sequential data matrix, i.e., $\mathbf{R}^A$, using our proposed random mechanism $\mathcal{M}$ to produce a perturbed sequential data matrix, i.e., $\tilde{\mathbf{R}}^A$.
This perturbed matrix is then transmitted to the target domain.
In the second stage, the target domain can take advantage of any existing SFCDSR algorithm, with the perturbed sequential data matrix serving as the input for the algorithm in place of the original sequential data matrix.
By doing this, the knowledge transfer from the auxiliary domain to the target domain can be achieved.
As shown in Fig. \ref{fig:pricdsr}, our PriCDSR is intuitively non-invasive, as it requires no modification to the base model of CDSR.
Therefore, our PriCDSR can be easily extended to the case of one target domain and multiple auxiliary domains.
The only requirement is to apply the SDP to the data from each auxiliary domain and transfer the resulting output to the target domain.
Note that our PriCDSR can only be applied to SFCDSR, not MFCDSR, because the algorithm of MFCDSR uses timestamps to merge the interaction records from the auxiliary domain and the target domain.
However, our PriCDSR does not consider the privacy protection of timestamps.

\begin{figure}[ht]
	\centering
	\includegraphics[width=\linewidth]{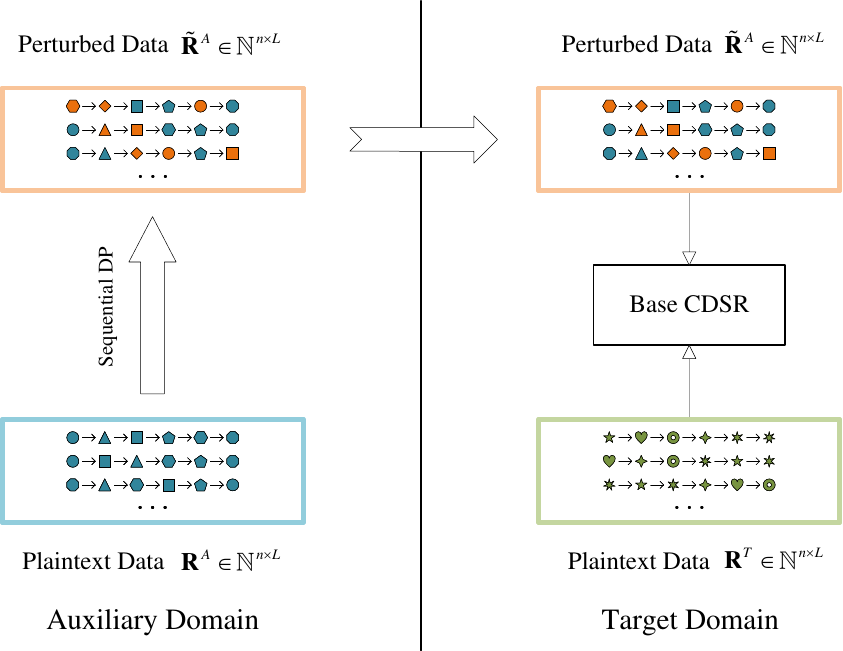}
	\caption{{Overview of our privacy-preserving cross-domain sequential recommender system (PriCDSR).
	The plaintext data of the auxiliary domain is represented as a sequential data matrix $\mathbf{R}^A$ (the blue box in the bottom left corner), which consists of the users' historical interaction sequences.
		$\mathbf{R}^A$ is added with noise to generate the perturbed sequential data matrix $\tilde{\mathbf{R}}^A$ (the orange box in the upper left corner).
		We incorporate these noises (the orange icons in the orange boxes) using the proposed random mechanism $\mathcal{M}$ (see Fig. \ref{alg}) satisfying Sequential DP (see Definition \ref{def:sdp}).
		Subsequently, the perturbed sequential data matrix of the auxiliary domain is transmitted to the target domain (the orange box in the upper right corner).
		The target domain can use any CDSR method as the base model, and legally use the perturbed sequential data matrix of the auxiliary domain $\tilde{\mathbf{R}}^A$ and the plaintext sequential data matrix of the target domain $\mathbf{R}^T$ (the green box in the bottom right corner) to provide recommendation services.}}
	\label{fig:pricdsr}
\end{figure}

Before introducing our random mechanism $\mathcal{M}$, we briefly introduce the randomized response (RR) algorithm~\cite{RR}.
Given a real item ID $v$, RR will output $i$ with the following probability:
$$
	\Pr[i|v] = \begin{cases}
		\frac{\exp(\epsilon)}{\exp(\epsilon)+m-1}, & i=v      \\
		\frac{1}{\exp(\epsilon)+m-1},              & i \neq v
	\end{cases}
$$
where $m$ is the number of all possible outputs.
And then, $i$ will be used as the noised data of $v$.

Our random mechanism $\mathcal{M}$ builds on the RR algorithm and improves upon it, which is shown in Fig. \ref{alg}.
Specifically, for each element $\mathbf{R}_{u,\ell}$ in the sequential data matrix, our random mechanism $\mathcal{M}$ first samples a candidate item ID $i$ using the probability distribution defined in (\ref{eq:pro}).
Note that our random mechanism $\mathcal{M}$ is designed to avoid sampling item IDs that has been sampled, i.e., the item IDs in $\mathcal{P}_{u,\ell}$, because an item ID can only appear once in a user's historical interaction sequence (please see Definition \ref{def:sdm}).
Subsequently, our random mechanism $\mathcal{M}$ check whether the item ID $i$ appears in the user's subsequent interacted item IDs.
If so, our random mechanism $\mathcal{M}$ swaps the subsequent item ID $i$ with the current item ID.
Otherwise, our random mechanism $\mathcal{M}$ replaces the current item ID with item ID $i$.

\begin{equation}
	\label{eq:pro}
	\begin{split}
		&\Pr[i|\mathcal{P}_{u,\ell},\mathbf{R}_{u,\ell}]\\
		=& \begin{cases}
			\frac{\exp(\epsilon)}{\exp(\epsilon)+m-|\mathcal{P}_{u,\ell}\backslash\{0\}|}, & i=\mathbf{R}_{u,\ell}                                                         \\
			\frac{1}{\exp(\epsilon)+m-|\mathcal{P}_{u,\ell}\backslash\{0\}|},              & i \neq \mathbf{R}_{u,\ell} \land i \notin \mathcal{P}_{u,\ell}\backslash\{0\} \\
			0,                                                                             & i \in \mathcal{P}_{u,\ell}\backslash\{0\}
		\end{cases}
	\end{split}
\end{equation}
where $\mathcal{P}_{u,\ell}=\{\tilde{\mathbf{R}}_{u,1}, \tilde{\mathbf{R}}_{u,2}, \ldots, \tilde{\mathbf{R}}_{u,\ell-1}\}$ and $i\in\{0,1,2,\ldots,m\}$.

\begin{figure}[h]
	\caption{Our random algorithm $\mathcal{M}$}
	\label{alg}
	\begin{algorithmic}[1]
		\FOR{$u\in\{1,2,\ldots,n\}$}
		\FOR{$\ell \in \{1,2,\ldots,L\}$}
		\STATE Pick a random item ID $i\in\{0,1,2,\ldots,m\}$ with probability in (\ref{eq:pro}).
		\IF{$i \neq \mathbf{R}_{u,\ell}$}
		\IF{$i \neq 0 \land i \in \{\mathbf{R}_{u,\ell+1}, \mathbf{R}_{u,\ell+2}, \ldots, \mathbf{R}_{u,L}\}$}
		\STATE Swap $\mathbf{R}_{u,\ell}$ and $\mathbf{R}_{u,\ell'}$, where $\mathbf{R}_{u,\ell'}=i \land \mathbf{R}_{u,\ell'} \in \{\mathbf{R}_{u,\ell+1}, \mathbf{R}_{u,\ell+2}, \ldots, \mathbf{R}_{u,L}\}$.
		\ELSE
		\STATE $\mathbf{R}_{u,\ell}\leftarrow i$.
		\ENDIF
		\ENDIF
		\ENDFOR
		\ENDFOR
	\end{algorithmic}
\end{figure}

Inspired by RR, we use a similar probability distribution (i.e., (\ref{eq:pro})) to sample item IDs, ensuring that our random mechanism $\mathcal{M}$ can be proven to satisfy $\epsilon$-SDP.
However, our random mechanism $\mathcal{M}$ is significantly different from RR.
It has the item IDs that cannot be sampled and the swap operations after getting the sampled item IDs.
Both of these differences are designed for sequential data.

Compared with PPGenCDR~\cite{PPGenCDR} and the work~\cite{eggs}, our random mechanism $\mathcal{M}$ provides theoretical privacy protection of the ID and order information.
Compared with PriCDR~\cite{PriCDR}, our random mechanism $\mathcal{M}$ provides privacy protection of sequential information.
Specifically, the JLT and SJLT used in PriCDR can only be applied to the rating matrix without sequential information.

Neither CLDP~\cite{CLDP} nor TLDP~\cite{TLDP} is a recommender system.
If they are applied to cross-domain sequential recommendation, the former can only protect the ID information, while the latter can only protect the sequential information.
Compared with them, our random mechanism $\mathcal{M}$ protects both the ID information and sequential information.

Although our random mechanism $\mathcal{M}$ seems simple, it does satisfy $\epsilon$-SDP, which will be proven in Section \ref{sec:priAna}.
Another simpler and seemingly effective way is to perturb the item IDs followed by random shuffling.
However, it is difficult to prove that this simpler method satisfies $\epsilon$-SDP.
The theoretical privacy guarantee firmly protects the privacy of the users, which may make the users more willing to participate in the training of recommendation models and meet the requirements of data protection laws and regulations.

The computational complexity of our random mechanism $\mathcal{M}$ is $O(n \times L)$.
It can be easily seen from Fig. \ref{alg} that the random mechanism $\mathcal{M}$ consists of two layers of cyclic structure.
The first layer loops $n$ times, and the second layer loops $L$ times.
As for the judgment in the 5th line, using a hash table can reduce its computational complexity to $O(1)$.
Specifically, before perturbing a user's interaction sequence (i.e., the second loop), traverse the interaction sequence to obtain a hash table whose keys are the item IDs and whose values are the positions of the item IDs in the sequence.
By doing so, the computational complexity of the entire algorithm can be reduced to $O(n \times L)$.
And this only requires an additional $O(L)$ spatial complexity.

An important feature of our mechanism is that the perturbed sequential data matrix generated by our random mechanism $\mathcal{M}$ has the same shape as the original sequential data matrix, i.e., $n \times L$.
This ensures the non-invasiveness of our PriCDSR.
Therefore, any existing SFCDSR algorithm can be applied to the perturbed sequential data matrix, thereby enabling the target domain to benefit from the transferred knowledge.

Now, we provide a detailed explanation of how our random mechanism $\mathcal{M}$ operates using a toy example, as shown in Fig. \ref{fig:mechanism}.
we consider a user's historical interaction sequence consisting of five items, i.e., $1 \rightarrow 2 \rightarrow 3 \rightarrow 4 \rightarrow 5$.
To conform to the fixed interaction sequence length, which is $8$ in our toy example, we zero-pad the sequence.
Our random mechanism $\mathcal{M}$ then perturbs each item ID (including the padded zeros) in the user's interaction sequence one by one.
Note that we perturb the padded zeros so that our random mechanism can satisfy $\epsilon$-SDP, although this will introduce more noise.
We explain the process below.
Firstly, for the first $0$, our random mechanism $\mathcal{M}$ samples $0$, so the sequence does not change.
Note that we do not illustrate this sampling in Fig. \ref{fig:mechanism} due to space limitation.
Secondly, for the second $0$, our random mechanism $\mathcal{M}$ samples item ID $8$, which is an item ID that does not appear in the original interaction sequence (we highlight this kind of item ID in orange in Fig. \ref{fig:mechanism}), so the item ID $8$ replaces the original $0$.
The result of this type of perturbation is similar to inserting a new item ID in the original interaction sequence.
Thirdly, for the third $0$, our random mechanism $\mathcal{M}$ samples item ID $1$, which is an item ID appears in the original interaction sequence (we highlight this kind of item ID in blue in Fig. \ref{fig:mechanism}).
Therefore, this $0$ and the item ID $1$ in the original sequence swap places.
Fourthly, for the swapped $0$ and item ID $2$, our random mechanism $\mathcal{M}$ samples $0$ and item ID $2$ respectively, so the sequence remains the same.
We also do not illustrate these two samplings in Fig. \ref{fig:mechanism}.
Fifthly, for the item ID $3$, our random mechanism $\mathcal{M}$ samples item ID $5$, and then swaps their positions in the interaction sequence.
Sixthly, for the item ID $4$, our random mechanism $\mathcal{M}$ samples $0$, so item ID $4$ is replaced with $0$.
The result of this type of perturbation is similar to deleting an item ID in the original interaction sequence.
Finally, for the swapped item ID $3$, our random mechanism $\mathcal{M}$ samples item ID $6$, an item ID that does not appear in the original sequence, so item ID $3$ is replaced with item ID $6$.
After all these perturbations, our random mechanism $\mathcal{M}$ transforms the original interaction sequence $0 \rightarrow 0 \rightarrow 0 \rightarrow 1 \rightarrow 2 \rightarrow 3 \rightarrow 4 \rightarrow 5$ into a perturbed sequence $0 \rightarrow 8 \rightarrow 1 \rightarrow 0 \rightarrow 2 \rightarrow 5 \rightarrow 0 \rightarrow 6$.

Note that this example is just an illustration of all the possible perturbations our random mechanism $\mathcal{M}$ can make to an interaction sequence.
Therefore, the perturbed interaction sequence is quite different from the original sequence.
This may not be the case in real-world scenarios.
The strength of our random mechanism $\mathcal{M}$ to perturb an interaction sequence is governed by the privacy budget $\epsilon$.

\begin{figure}[ht]
	\centering
	\includegraphics[width=\linewidth]{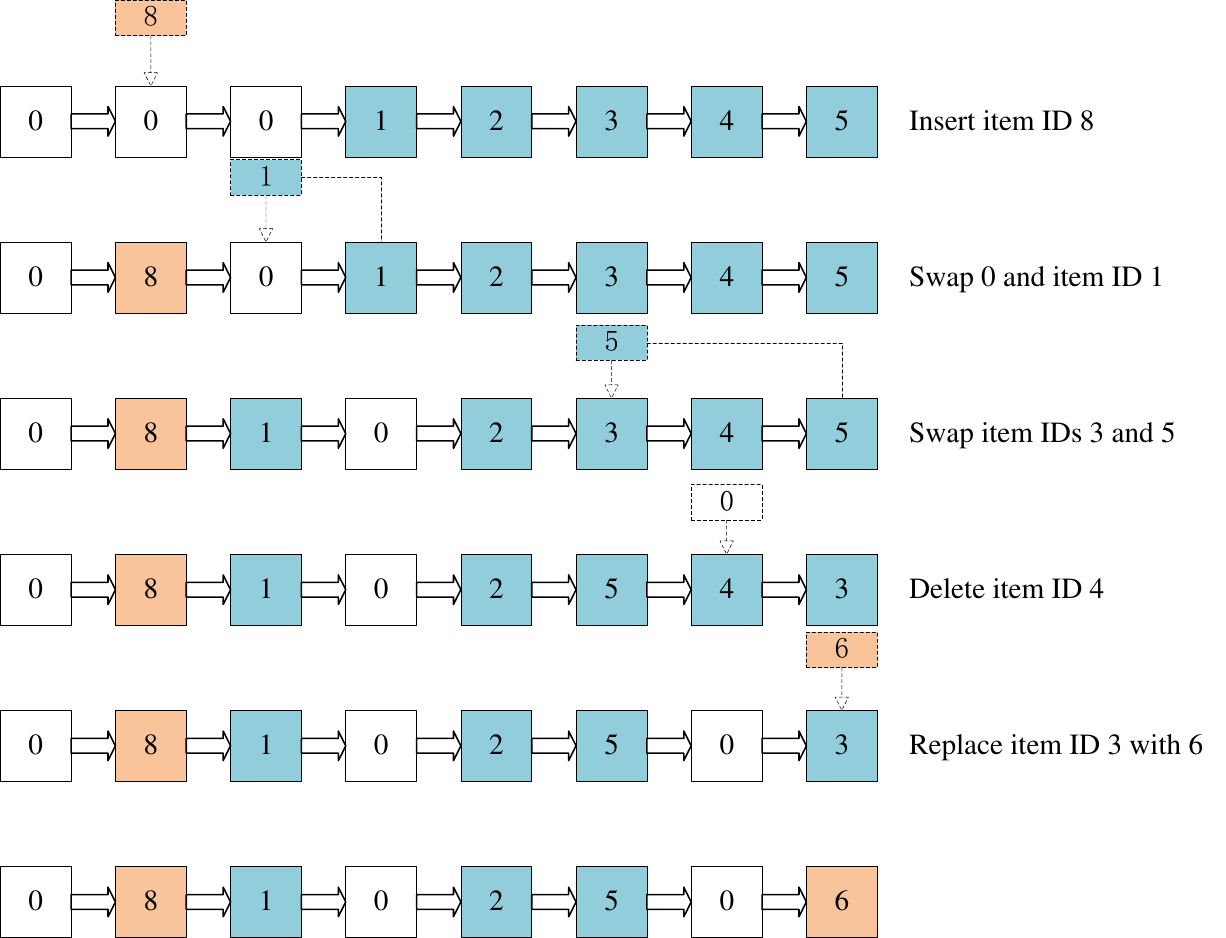}
	\caption{Illustration of an example of our random mechanism $\mathcal{M}$.
		From top to bottom, each interaction sequence shows each step of how our random mechanism $\mathcal{M}$ adds noise to the example interaction sequence.
		The item ID sampled at each step is represented by a dotted box (the dotted arrow extending from it points to the current item ID (or padded zero)).
		If the sampled item ID appears in the subsequent subsequence (connected by a dotted line), the current item ID (or padded zero) and the subsequent item ID will be swapped.
		Otherwise, the current item ID is replaced with the sampled item ID.
		Based on the sampled item ID, the current item ID, and whether the sampled item ID appears in the subsequent subsequence, the added noise can be classified into four categories, i.e., insertion, deletion, replacement, and swap.
		The specific operation of each step is explained on the right side of the corresponding sequence.}
	\label{fig:mechanism}
\end{figure}

\section{Privacy Analysis}
\label{sec:priAna}
In this section, we prove that our random mechanism $\mathcal{M}$ satisfies $\epsilon$-SDP, as shown in Theorem \ref{the}.
Note that the random mechanism $\mathcal{M}$ satisfying SDP is only applied to the data in the auxiliary domain.
Therefore, we will describe the details of our random mechanism $\mathcal{M}$ using the auxiliary domain data as input.
To keep the subsequent description concise, we will use $\mathcal{I}$, $m$, $\mathbf{R}$ and $\tilde{\mathbf{R}}$ to refer to $\mathcal{I}^A$, $m_A$, $\mathbf{R}^A$ and $\tilde{\mathbf{R}}^A$, respectively, unless stated otherwise.

\begin{theorem}
	\label{the}
	Our random algorithm $\mathcal{M}$ guarantees $\epsilon$-SDP.
\end{theorem}

\begin{proof}
	According to Definition \ref{def:nsdm}, we divide the proof into Case 1 and Case 2.
	Case 1 deals with the situation where the neighbouring sequential data matrices differ from the value of a single element, while Case 2 considers the scenario where they differ from the position of a pair of elements.

	\textbf{Case 1}.
	For two neighbouring sequential data matrices $\mathbf{R}$ and $\mathbf{R}'$, which differ only by one user-item interaction, and any output $\mathbf{A}$, we expand $\frac{\Pr[\mathcal{M}(\mathbf{R})=\mathbf{A}]}{\Pr[\mathcal{M}(\mathbf{R}')=\mathbf{A}]}$ using the conditional probability formula as follows.
	\begin{equation}
		\label{eq:case1}
		\begin{split}
			\frac{\Pr[\mathcal{M}(\mathbf{R})=\mathbf{A}]}{\Pr[\mathcal{M}(\mathbf{R}')=\mathbf{A}]}
			=    & \frac{\prod_{u=1}^n\prod_{\ell=1}^L\Pr[\tilde{\mathbf{R}}_{u,\ell}=\mathbf{A}_{u,\ell}|\mathcal{P}_{u,\ell}]}
			{\prod_{u=1}^n\prod_{\ell=1}^L\Pr[\tilde{\mathbf{R}}'_{u,\ell}=\mathbf{A}_{u,\ell}|\mathcal{P}_{u,\ell}]}                   \\
			=    & \frac{\prod_{\ell=1}^{\ell_0-1}\Pr[\tilde{\mathbf{R}}_{u_0,\ell}=\mathbf{A}_{u_0,\ell}|\mathcal{P}_{u_0,\ell}]}
			{\prod_{\ell=1}^{\ell_0-1}\Pr[\tilde{\mathbf{R}}'_{u_0,\ell}=\mathbf{A}_{u_0,\ell}|\mathcal{P}_{u_0,\ell}]}\\
			*    & \frac{\Pr[\tilde{\mathbf{R}}_{u_0,\ell_0}=\mathbf{A}_{u_0,\ell_0}|\mathcal{P}_{u_0,\ell_0}]}
			{\Pr[\tilde{\mathbf{R}}'_{u_0,\ell_0}=\mathbf{A}_{u_0,\ell_0}|\mathcal{P}_{u_0,\ell_0}]}\\
			*    & \frac{\prod_{\ell=\ell_0+1}^{L}\Pr[\tilde{\mathbf{R}}_{u_0,\ell}=\mathbf{A}_{u_0,\ell}|\mathcal{P}_{u_0,\ell}]}
			{\prod_{\ell=\ell_0+1}^{L}\Pr[\tilde{\mathbf{R}}'_{u_0,\ell}=\mathbf{A}_{u_0,\ell}|\mathcal{P}_{u_0,\ell}]}\\
		\end{split}
	\end{equation}
	For any $\mathbf{A}_{u_0,\ell}, 1 \leq \ell < \ell_0$, $\Pr[\tilde{\mathbf{R}}_{u_0,\ell}=\mathbf{A}_{u_0,\ell}|\mathcal{P}_{u_0,\ell}]$ and $\Pr[\tilde{\mathbf{R}}'_{u_0,\ell}=\mathbf{A}_{u_0,\ell}|\mathcal{P}_{u_0,\ell}]$ are obviously equal no matter what $\mathbf{A}_{u_0,\ell}$ is.
	During the process of perturbing $\mathbf{R}_{u_0,\ell}$ and $\mathbf{R}'_{u_0,\ell}$, $1 \leq \ell < \ell_0$, it is possible that the item IDs of subsequent elements may also be perturbed, because our random mechanism $\mathcal{M}$ may sample an item ID that exists in the subsequent elements and swap its position with the current one.
	However, this is acceptable because it does not change the equality of each element of $\mathbf{R}_{u_0,\ell}$ and $\mathbf{R}'_{u_0,\ell}$, $\ell_0 \leq \ell \leq L$.
	Specifically, for $\ell=\ell_0$, no matter what the currently sampled item ID is, it can only be swapped with $\mathbf{R}_{u_0,\ell_0}$ or $\mathbf{R}'_{u_0,\ell_0}$, and the new $\mathbf{R}_{u_0,\ell_0}$ or $\mathbf{R}'_{u_0,\ell_0}$ cannot be equal to the $\mathbf{R}'_{u_0,\ell_0}$ or $\mathbf{R}_{u_0,\ell_0}$ that has not been swapped. 
	Therefore, $\mathbf{R}_{u_0,\ell_0}$ and $\mathbf{R}'_{u_0,\ell_0}$ remain unequal after the possible swap.
	For $\ell_0 < \ell \leq L$, $\mathbf{R}_{u_0,\ell}$ and $\mathbf{R}'_{u_0,\ell}$ will be swapped together or not because they are equal. 
	Hence, $\mathbf{R}_{u_0,\ell}$ and $\mathbf{R}'_{u_0,\ell}$, $\ell_0 < \ell \leq L$, remain equal after the possible swap.
	Therefore, Equation (\ref{eq:case1}) can be rewritten as follows.
	\begin{equation}
		\label{eq:case1-simple}
		\begin{split}
			\frac{\Pr[\mathcal{M}(\mathbf{R})=\mathbf{A}]}{\Pr[\mathcal{M}(\mathbf{R}')=\mathbf{A}]}
			=    & \frac{\Pr[\tilde{\mathbf{R}}_{u_0,\ell_0}=\mathbf{A}_{u_0,\ell_0}|\mathcal{P}_{u_0,\ell_0}]}
			{\Pr[\tilde{\mathbf{R}}'_{u_0,\ell_0}=\mathbf{A}_{u_0,\ell_0}|\mathcal{P}_{u_0,\ell_0}]} \\
			*    & \frac{\prod_{\ell=\ell_0+1}^{L}\Pr[\tilde{\mathbf{R}}_{u_0,\ell}=\mathbf{A}_{u_0,\ell}|\mathcal{P}_{u_0,\ell}]}
			{\prod_{\ell=\ell_0+1}^{L}\Pr[\tilde{\mathbf{R}}'_{u_0,\ell}=\mathbf{A}_{u_0,\ell}|\mathcal{P}_{u_0,\ell}]} \\
		\end{split}
	\end{equation}

	For this equation, we consider the following four cases.

	\textbf{Case 1.1}:
	If the candidate item ID of our random mechanism $\mathcal{M}$ to $\mathbf{R}_{u_0,\ell_0}$ and $\mathbf{R}'_{u_0,\ell_0}$ is $\mathbf{R}_{u_0,\ell_0}$, i.e., $\mathbf{A}_{u_0,\ell_0} = \mathbf{R}_{u_0,\ell_0}$, then the random mechanism $\mathcal{M}$ will assign $\mathbf{R}_{u_0,\ell_0}$ to $\mathbf{R}'_{u_0,\ell_0}$.
	After this operation, the subsequent elements, i.e., $\mathbf{R}_{u_0,\ell}$ and $\mathbf{R}'_{u_0,\ell}, \ell_0 < \ell \leq L$, remains the same, and are still point-wise equal.
	Hence, $\Pr[\tilde{\mathbf{R}}_{u_0,\ell}=\mathbf{A}_{u_0,\ell}|\mathcal{P}_{u_0,\ell}]$ and $\Pr[\tilde{\mathbf{R}}'_{u_0,\ell}=\mathbf{A}_{u_0,\ell}|\mathcal{P}_{u_0,\ell}]$ are equal no matter what $\mathbf{A}_{u_0,\ell}, \ell_0+1 \leq \ell \leq L$ is. Therefore, Equation (\ref{eq:case1-simple}) can be calculated as follows.
	\begin{equation}
		\frac{\Pr[\mathcal{M}(\mathbf{R})=\mathbf{A}]}{\Pr[\mathcal{M}(\mathbf{R}')=\mathbf{A}]}
		=  \frac{\frac{\exp(\epsilon)}{\exp(\epsilon)+m-|\mathcal{P}_{u_0,\ell_0}\backslash\{0\}|}}{\frac{1}{\exp(\epsilon)+m-|\mathcal{P}_{u_0,\ell_0}\backslash\{0\}|}} * 1  = \exp(\epsilon)
	\end{equation}

	\textbf{Case 1.2}:
	If the candidate item ID of our random mechanism $\mathcal{M}$ to $\mathbf{R}_{u_0,\ell_0}$ and $\mathbf{R}'_{u_0,\ell_0}$ is $\mathbf{R}'_{u_0,\ell_0}$, i.e., $\mathbf{A}_{u_0,\ell_0} = \mathbf{R}'_{u_0,\ell_0}$.
	Hence, Equation (\ref{eq:case1-simple}) can be calculated as follows.
	\begin{equation}
		\frac{\Pr[\mathcal{M}(\mathbf{R})=\mathbf{A}]}{\Pr[\mathcal{M}(\mathbf{R}')=\mathbf{A}]}
		=  \frac{\frac{1}{\exp(\epsilon)+m-|\mathcal{P}_{u_0,\ell_0}\backslash\{0\}|}}{\frac{\exp(\epsilon)}{\exp(\epsilon)+m-|\mathcal{P}_{u_0,\ell_0}\backslash\{0\}|}} * 1  =  \frac{1}{\exp(\epsilon)}
	\end{equation}
	\textbf{Case 1.3}: If the candidate item ID of our random mechanism $\mathcal{M}$ to $\mathbf{R}_{u_0,\ell_0}$ and $\mathbf{R}'_{u_0,\ell_0}$ is an item ID that exists in the subsequent elements, i.e., $\mathbf{A}_{u_0,\ell_0} \neq \mathbf{R}_{u_0,\ell_0} \land \mathbf{A}_{u_0,\ell_0} \neq \mathbf{R}'_{u_0,\ell_0} \land \mathbf{A}_{u_0,\ell_0} \in \{\mathbf{R}_{u_0,\ell_0+1}, \mathbf{R}_{u_0,\ell_0+2}, \ldots, \mathbf{R}_{u_0,L}\}$, then our random mechanism $\mathcal{M}$ will swap $\mathbf{R}_{u_0,\ell_0}$ with the sampled item ID.
	In this case, $\frac{\Pr[\tilde{\mathbf{R}}_{u_0,\ell_0}=\mathbf{A}_{u_0,\ell_0}|\mathcal{P}_{u_0,\ell_0}]}{\Pr[\tilde{\mathbf{R}}'_{u_0,\ell_0}=\mathbf{A}_{u_0,\ell_0}|\mathcal{P}_{u_0,\ell_0}]} = 1$, and $\mathbf{R}_{u_0,\ell}$ and $\mathbf{R}'_{u_0,\ell}, \ell_0 < \ell \leq L$ become two new neighbouring sequential data matrices that differ from the value of a single element, i.e., back to Case 1.

	\textbf{Case 1.4}: If the candidate item ID of our random mechanism $\mathcal{M}$ to $\mathbf{R}_{u_0,\ell_0}$ and $\mathbf{R}'_{u_0,\ell_0}$ is a new item ID, i.e., $\mathbf{A}_{u_0,\ell_0} \neq \mathbf{R}_{u_0,\ell_0} \land \mathbf{A}_{u_0,\ell_0} \neq \mathbf{R}'_{u_0,\ell_0} \land \mathbf{A}_{u_0,\ell_0} \notin \{\mathbf{R}_{u_0,\ell_0+1}, \mathbf{R}_{u_0,\ell_0+1}, \ldots, \mathbf{R}_{u_0,L}\}$, then our random mechanism $\mathcal{M}$ will assign the new item ID to $\mathbf{R}_{u_0,\ell_0}$ and $\mathbf{R}'_{u_0,\ell_0}$.
	In this case, $\frac{\Pr[\tilde{\mathbf{R}}_{u_0,\ell_0}=\mathbf{A}_{u_0,\ell_0}|\mathcal{P}_{u_0,\ell_0}]}{\Pr[\tilde{\mathbf{R}}'_{u_0,\ell_0}=\mathbf{A}_{u_0,\ell_0}|\mathcal{P}_{u_0,\ell_0}]} = 1$, and $\mathbf{R}_{u_0,\ell}$ and $\mathbf{R}'_{u_0,\ell}, \ell_0 < \ell \leq L$ have not been changed so they are still point-wise equal.
	Therefore, Equation (\ref{eq:case1-simple}) can be calculated as follows.
	\begin{equation}
		\frac{\Pr[\mathcal{M}(\mathbf{R})=\mathbf{A}]}{\Pr[\mathcal{M}(\mathbf{R}')=\mathbf{A}]}
		=  1 * 1     =1
	\end{equation}

	To sum up, in Case 1, $\frac{\Pr[\mathcal{M}(\mathbf{R})=\mathbf{A}]}{\Pr[\mathcal{M}(\mathbf{R}')=\mathbf{A}]}\leq \exp(\epsilon)$.

	\textbf{Case 2}.
	For two neighbouring sequential data matrices $\mathbf{R}$ and $\mathbf{R}'$ in which only the positions of a pair of elements have been swapped, and any output $\mathbf{A}$,
	\begin{equation}
		\label{eq:case2}
		\begin{split}
			\frac{\Pr[\mathcal{M}(\mathbf{R})=\mathbf{A}]}{\Pr[\mathcal{M}(\mathbf{R}')=\mathbf{A}]}
			= & \frac{\prod_{u=1}^n\prod_{\ell=1}^L\Pr[\tilde{\mathbf{R}}_{u,\ell}=\mathbf{A}_{u,\ell}|\mathcal{P}_{u,\ell}]}
			{\prod_{u=1}^n\prod_{\ell=1}^L\Pr[\tilde{\mathbf{R}}'_{u,\ell}=\mathbf{A}_{u,\ell}|\mathcal{P}_{u_1,\ell}]}\\
			= & \frac{\prod_{\ell=1}^{\ell_1-1}\Pr[\tilde{\mathbf{R}}_{u_1,\ell}=\mathbf{A}_{u_1,\ell}|\mathcal{P}_{u_1,\ell}]}
			{\prod_{\ell=1}^{\ell_1-1}\Pr[\tilde{\mathbf{R}}'_{u_1,\ell}=\mathbf{A}_{u_1,\ell}|\mathcal{P}_{u_1,\ell}]}\\
			* & \frac{\Pr[\tilde{\mathbf{R}}_{u_1,\ell_1}=\mathbf{A}_{u_1,\ell_1}|\mathcal{P}_{u_1,\ell_1}]}
			{\Pr[\tilde{\mathbf{R}}'_{u_1,\ell_1}=\mathbf{A}_{u_1,\ell_1}|\mathcal{P}_{u_1,\ell_1}]}\\
			* & \frac{\prod_{\ell=\ell_1+1}^{L}\Pr[\tilde{\mathbf{R}}_{u_1,\ell}=\mathbf{A}_{u_1,\ell}|\mathcal{P}_{u_1,\ell}]}
			{\prod_{\ell=\ell_1+1}^{L}\Pr[\tilde{\mathbf{R}}'_{u_1,\ell}=\mathbf{A}_{u_1,\ell}|\mathcal{P}_{u_1,\ell}]}
		\end{split}
	\end{equation}

	For any $1 \leq \ell < \ell_1$, if our random algorithm $\mathcal{M}$ swaps $\mathbf{R}_{u_1,\ell}$ and $\mathbf{R}_{u_1,\ell_1}$ (or $\mathbf{R}_{u_1,\ell}$ and $\mathbf{R}_{u_1,\ell_2}$), then it must swap $\mathbf{R}'_{u_1,\ell}$ and $\mathbf{R}'_{u_1,\ell_2}$ (or $\mathbf{R}'_{u_1,\ell}$ and $\mathbf{R}'_{u_1,\ell_1}$).
	This does not change the equality of each element of $\mathbf{R}_{u_1,\ell}$ and $\mathbf{R}'_{u_1,\ell}$, where $\ell_1 \leq \ell \leq L$.
	Therefore, Equation (\ref{eq:case2}) can be rewritten as follows.
	\begin{equation}
		\label{eq:case2-simple}
		\begin{split}
			\frac{\Pr[\mathcal{M}(\mathbf{R})=\mathbf{A}]}{\Pr[\mathcal{M}(\mathbf{R}')=\mathbf{A}]}
			= & \frac{\Pr[\tilde{\mathbf{R}}_{u_1,\ell_1}=\mathbf{A}_{u_1,\ell_1}|\mathcal{P}_{u_1,\ell_1}]}
			{\Pr[\tilde{\mathbf{R}}'_{u_1,\ell_1}=\mathbf{A}_{u_1,\ell_1}|\mathcal{P}_{u_1,\ell_1}]}\\
			* & \frac{\prod_{\ell=\ell_1+1}^{L}\Pr[\tilde{\mathbf{R}}_{u_1,\ell}=\mathbf{A}_{u_1,\ell}|\mathcal{P}_{u_1,\ell}]}
			{\prod_{\ell=\ell_1+1}^{L}\Pr[\tilde{\mathbf{R}}'_{u_1,\ell}=\mathbf{A}_{u_1,\ell}|\mathcal{P}_{u_1,\ell}]}\\
		\end{split}
	\end{equation}

	For this equation, we consider the following four cases.

	\textbf{Case 2.1}: If the candidate item ID of our random mechanism $\mathcal{M}$ to $\mathbf{R}_{u_1,\ell_1}$ and $\mathbf{R}'_{u_1,\ell_1}$ is $\mathbf{R}_{u_1,\ell_1}$, i.e., $\mathbf{A}_{u_1,\ell_1} = \mathbf{R}_{u_1,\ell_1}=\mathbf{R}'_{u_1,\ell_2}$, then the random mechanism $\mathcal{M}$ will swap $\mathbf{R}'_{u_1,\ell_1}$ and $\mathbf{R}'_{u_1,\ell_2}$.
	After this operation, the subsequent elements in the neighbouring sequential data matrices are point-wise equal, i.e., $\mathbf{R}_{u_1,\ell} = \mathbf{R}'_{u_1,\ell}, \ell_1 < \ell \leq L$.
	Therefore, Equation (\ref{eq:case2-simple}) can be calculated as follows.
	\begin{equation}
		\frac{\Pr[\mathcal{M}(\mathbf{R})=\mathbf{A}]}{\Pr[\mathcal{M}(\mathbf{R}')=\mathbf{A}]}
		=  \frac{\frac{\exp(\epsilon)}{\exp(\epsilon)+m-|\mathcal{P}_{u_1,\ell_1}\backslash\{0\}|}}{\frac{1}{\exp(\epsilon)+m-|\mathcal{P}_{u,\ell_1}\backslash\{0\}|}} * 1  
		=  \exp(\epsilon)
	\end{equation}

	\textbf{Case 2.2}: If the candidate item ID of our random mechanism $\mathcal{M}$ to $\mathbf{R}_{u_1,\ell_1}$ and $\mathbf{R}'_{u_1,\ell_1}$ is $\mathbf{R}'_{u_1,\ell_1}$, i.e., $\mathbf{A}_{u_1,\ell_1} = \mathbf{R}'_{u_1,\ell_1} = \mathbf{R}_{u_1,\ell_2}$, then the random mechanism $\mathcal{M}$ will swap $\mathbf{R}_{u_1,\ell_1}$ and $\mathbf{R}_{u_1,\ell_2}$.
	After this operation, the subsequent elements in the neighbouring sequential data matrices are also point-wise equal, just like Case 2.1.
	Hence, Equation (\ref{eq:case2-simple}) can be calculated as follows.
	\begin{equation}
			\frac{\Pr[\mathcal{M}(\mathbf{R})=\mathbf{A}]}{\Pr[\mathcal{M}(\mathbf{R}')=\mathbf{A}]}
			=  \frac{\frac{1}{\exp(\epsilon)+m-|\mathcal{P}_{u_1,\ell_1}\backslash\{0\}|}}{\frac{\exp(\epsilon)}{\exp(\epsilon)+m-|\mathcal{P}_{u,\ell_1}\backslash\{0\}|}} * 1  
			=  \frac{1}{\exp(\epsilon)}
	\end{equation}

	\textbf{Case 2.3}: If the candidate item ID of our random mechanism $\mathcal{M}$ to $\mathbf{R}_{u_1,\ell_1}$ and $\mathbf{R}'_{u_1,\ell_1}$ is $\mathbf{R}_{u_1,\ell}$, where $\ell_1 < \ell \leq L \land \ell \neq \ell_2$, i.e., $\mathbf{A}_{u_1,\ell_1} \neq \mathbf{R}_{u_1,\ell_1} \land \mathbf{A}_{u_1,\ell_1} \neq \mathbf{R}'_{u_1,\ell_1} \land \mathbf{A}_{u_1,\ell_1} \in \{\mathbf{R}_{u_1,\ell_1+1}, \ldots, \mathbf{R}_{u_1,\ell_2-1}, \mathbf{R}_{u_1,\ell_2+1}, \ldots, \mathbf{R}_{u_1,L}\}$, then the random mechanism $\mathcal{M}$ will swap $\mathbf{R}_{u_1,\ell_1}$ with the item ID $\mathbf{R}_{u_1,\ell}$.
	In this case, $\frac{\Pr[\tilde{\mathbf{R}}_{u_1,\ell_1}=\mathbf{A}_{u_1,\ell_1}|\mathcal{P}_{u_1,\ell_1}]}
		{\Pr[\tilde{\mathbf{R}}'_{u_1,\ell_1}=\mathbf{A}_{u_1,\ell_1}|\mathcal{P}_{u_1,\ell_1}]} = 1$, and $\mathbf{R}_{u_1,\ell}$ and $\mathbf{R}'_{u_1,\ell}, \ell_1 < \ell \leq L$ become two new neighbouring sequential data matrices that differ from the position of a pair of elements, i.e., back to Case 2.

	\textbf{Case 2.4}: If the candidate item ID of our random mechanism $\mathcal{M}$ to $\mathbf{R}_{u_1,\ell_1}$ and $\mathbf{R}'_{u_1,\ell_1}$ is a new item ID, i.e., $\mathbf{A}_{u_1,\ell_1} \neq \mathbf{R}_{u_1,\ell_1} \land \mathbf{A}_{u_1,\ell_1} \neq \mathbf{R}'_{u_1,\ell_1} \land \mathbf{A}_{u_1,\ell_1} \notin \{\mathbf{R}_{u_1,\ell_1+1}, \ldots, \mathbf{R}_{u_1,\ell_2-1}, \mathbf{R}_{u_1,\ell_2+1}, \ldots, \mathbf{R}_{u_1,L}\}$, then our random mechanism $\mathcal{M}$ will assign the new item ID to $\mathbf{R}_{u_1,\ell_1}$ and $\mathbf{R}'_{u_1,\ell_1}$.
	In this case, $\frac{\Pr[\tilde{\mathbf{R}}_{u_1,\ell_1}=\mathbf{A}_{u_1,\ell_1}|\mathcal{P}_{u_1,\ell_1}]}{\Pr[\tilde{\mathbf{R}}'_{u_1,\ell_1}=\mathbf{A}_{u_1,\ell_1}|\mathcal{P}_{u_1,\ell_1}]} = 1$, and $\mathbf{R}_{u_1,\ell}$ and $\mathbf{R}'_{u_1,\ell}, \ell_1 < \ell \leq L$ become two new neighbouring sequential data matrices that differ from the value of a single element, i.e., $\mathbf{R}_{u_1,\ell_2} \neq \mathbf{R}'_{u_1,\ell_2}$.
	Therefore, this case can be regarded as Case 1.
	Furthermore, in Case 1, $\frac{\Pr[\mathcal{M}(\mathbf{R})=\mathbf{A}]}{\Pr[\mathcal{M}(\mathbf{R}')=\mathbf{A}]}\leq \exp(\epsilon)$.
	Hence,
	\begin{equation}
		\frac{\Pr[\mathcal{M}(\mathbf{R})=\mathbf{A}]}{\Pr[\mathcal{M}(\mathbf{R}')=\mathbf{A}]} \leq 1 * \exp(\epsilon) = \exp(\epsilon)
	\end{equation}

	To sum up, in Case 2, $\frac{\Pr[\mathcal{M}(\mathbf{R})=\mathbf{A}]}{\Pr[\mathcal{M}(\mathbf{R}')=\mathbf{A}]}\leq \exp(\epsilon)$.

	Combining Case 1 and Case 2, Theorem \ref{the} is proved.
\end{proof}

\section{Empirical Evaluations}
\label{sec:exp}
In this section, we focus on the following three research questions (RQs) and conduct two corresponding experiments.
RQ1: How does our PriCDSR perform compared to a model trained only on the data of the target domain?
RQ2: How does our PriCDSR perform compared to a model trained on the plaintext data of both the target domain and the auxiliary domain?
RQ3: How does the privacy budget $\epsilon$ affect the performance of our PriCDSR?
\subsection{Data and Evaluation Metrics}

We use real-world data from Amazon~\cite{Amazon}, which is an e-commerce dataset collected by Amazon with item ratings and reviews, etc.
We choose three subsets, Movies and TV (Movie), CDs and Vinyl (CD), and Books (Book), as three domains for experiments.
We follow MGCL~\cite{MGCL} and preprocess these three domains as follows:
1) We assume that the presence of reviews, check-ins, and purchases is positive.
In other words, a user's interaction record with an item implies that the user is interested in the item.
2) We only keep the users and items with at least five interaction records.
3) We only keep the interaction sequences of users who have interaction records in all the three domains.
4) We use timestamps to determine the order of interactions and discard later duplicated interaction, i.e., user-item pairs.
5) We adopt leave-one-out evaluation by splitting each dataset into three parts, i.e., the last interaction of each user for testing, the penultimate one for validation, and the remaining interaction records for training.
The statistics of the preprocessed data are shown in Table \ref{tab:statis}.
We release the source codes and scripts to reproduce all the experimental results at \url{https://github.com/LachlanLin/PriCDSR}.

\begin{table}[htbp]
	\caption{Statistics of the data in three different domains.}
	\label{tab:statis}
    \centering
	\begin{tabular}{lrrr}
		\hline
		Domain & \#Users & \#Items & Avg. Seqlen. \\
		Movie  & 10,929  & 59,513  & 42.11        \\
		CD     & 10,929  & 91,169  & 31.50        \\
		Book   & 10,929  & 236,049 & 55.60        \\
		\hline
	\end{tabular}
\end{table}

Our research problem assumes that there is only one target domain and one auxiliary domain, hence we conduct experiments on all six possible (target domain, auxiliary domain) pairs, i.e., Movie$\leftarrow$Book, Movie$\leftarrow$CD, Book$\leftarrow$Movie, Book$\leftarrow$CD, CD$\leftarrow$Movie, and CD$\leftarrow$Book.

To evaluate the recommendation performance, we follow the literatures~\cite{SASRec, FISSA} and sample negative items based on their popularity.
For each user, we sample 100 items that have not been interacted with as negative items, and then rank these negative items with the ground-truth item.
We choose two commonly used metrics for CDSR, i.e., normalized discounted cumulative gain (NDCG@$k$) and hit ratio (HR@$k$).
Note that $k$ denotes the length of the recommendation list provided by the algorithms.

\subsection{Baselines and Hyperparameter Configurations}

The random mechanism $\mathcal{M}$ acts on the data rather than on the model (as illustrated in Fig. \ref{fig:pricdsr}), which means that our empirical studies shall focus on the data with and without applying our mechanism.
Therefore, we use a CDSR method, i.e., DASL~\cite{DASL}, as our base model and obtain an implementation of our PriCDSR (PriCDSR-DASL).
Note that we use the source code provided by the DASL authors at \url{https://github.com/lpworld/DASL}.
To enable DASL to be applied to our data, we make the following modifications to it.
(i) We modify the loss function.
The data used in the original paper has two kinds of labels, i.e., positive and negative.
Therefore, its loss function is the binary cross-entropy loss defined on the positive items and the negative items.
However, the data in our research problem only consists of one-class feedback.
We thus treat all interaction records as positive and sample negative items to define the binary cross-entropy loss.
(ii) We modify the network structure of the output layer, from the concatenation and fully connected layers in the original paper to the inner product.
Note that we do not simulate the data transfer between the two domains as shown in Fig. \ref{fig:pricdsr}.
Instead, we use our random mechanism $\mathcal{M}$ to add noise to the auxiliary data of each domain couple to obtain the perturbed auxiliary data.
We refer to the DASL using perturbed auxiliary data as our PriCDSR-DASL.
Therefore, the base model of our PriCDSR-DASL, i.e., DASL, remains unchanged and does not involve any modifications.
This implementation detail further demonstrates the non-invasiveness of our PriCDSR.

We adopt DASL without the DE and DA modules (DASL-single) as our baseline for single-domain sequential recommendation.
By removing the DE and DA modules, DASL can only rely on the data from the target domain to capture user interests.
It is important to note that we do not include other single-domain sequential recommendation models, such as SASRec~\cite{SASRec}, as our baselines.
The reason is that we focus on evaluating the extent to which our PriCDSR model affects the accuracy of the base CDSR model.
By comparing our PriCDSR-DASL with DASL-single, which is degenerated from the CDSR method DASL, we can ascertain if the noise introduced by our random mechanism $\mathcal{M}$ fully mitigates the performance gains arising from the auxiliary data.
Hence, comparing with a single-domain recommendation method that is not reduced from the cross-domain recommendation model is not necessary.

We search all the hyperparameters according to the NDCG@10 metric on the validation set of each data.
For DASL and DASL-single, we fix the number of epochs to 1000, the hidden size of the embedding layer to 128, the memory window of the attention layer to 10, the dropout rate of the dual attention layer to 0.5, and use Adam~\cite{Adam} as the optimizer.
We search the best values of the learning rate in $\{0.0001, 0.0005, 0.001, 0.005\}$, and the batch size in $\{128, 256, 512\}$.
As a result, the best learning rate is 0.001 and the best batch size is 128 on each data for both DASL and DASL-single.
For our PriCDSR-DASL, we do not search for hyperparameters on the validation set again, and instead use the same hyperparameters as that of DASL.
In order to ensure that the experimental results are reliable, we will report the average performance on the test set of three runs.

\subsection{Model Comparison (RQ1 and RQ2)}

To answer RQ1 and RQ2, we evaluate DASL, DASL-single and PriCDSR-DASL on each data.
We fix the privacy budget $\epsilon$ to 10 for PriCDSR-DASL.
We report the experimental results in Table \ref{tab:result}, from which we can have the following observations.
(1)
Compared with DASL-single, DASL has better recommendation performance on each data, which has been demonstrated in the original paper of DASL~\cite{DASL}.
This can be attributed to the DE and DA modules, which effectively leverage the auxiliary domain knowledge to enhance the recommendation performance in the target domain.
(2)
Compared with DASL-single, PriCDSR-DASL still has better performance even though it uses the perturbed data of the auxiliary domain.
This shows that our random mechanism $\mathcal{M}$ can still enable the DE and DA modules of DASL to use the data of the auxiliary domain to help the target domain improve the recommendation performance, while protecting the users' privacy.
Note that our PriCDSR does not help improve the recommendation performance. 
Our focus is on the privacy protection.
Thanks to the non-invasiveness of our PriCDSR, it is able to keep using the state-of-the-art model in CDSR as the base model to obtain higher recommendation performance.
This can reduce the cost for enterprises to update and iterate they recommender systems.
(3)
PriCDSR-DASL may perform better or worse than DASL depending on the data.
Specifically, PriCDSR-DASL performs worse than DASL on data Movie$\leftarrow$Book, Movie$\leftarrow$CD, Book$\leftarrow$Movie, and Book$\leftarrow$CD, but performs better on data CD$\leftarrow$Movie and CD$\leftarrow$Book.
The observed improvement or degradation in recommendation performance can be attributed to the noise introduced by our random mechanism $\mathcal{M}$.
On the one hand, the noise may disrupt the knowledge transferred from the auxiliary domain, thereby reducing the benefit to the target domain.
On the other hand, the introduction of a suitable amount of noise to the auxiliary domain data may enhance the robustness of the target domain model and, in turn, lead to improved recommendation performance.

\begin{table*}[htbp]
	\caption{Recommendation performance of the single-domain sequential recommendation (SR) method DASL-single, cross-domain SR (CDSR) method DASL, and our privacy-preserving CDSR method PriCDSR-DASL on each data.
		Note that we mainly focus on comparison of the data with or without using the proposed random mechanism $\mathcal{M}$.
	}
	\label{tab:result}
	\centering
	\begin{tabular}{lllcccc}
		\hline
		Datasets & Algorithms   & Mechanism         & HR@5                     & HR@10                    & NDCG@5                   & NDCG@10                  \\
		\hline
		\multirow{3}{*}{Movie$\leftarrow$Book}
		         & DASL-single  & w/o $\mathcal{M}$ & 0.1754{\tiny$\pm$0.0013} & 0.2701{\tiny$\pm$0.0040} & 0.1164{\tiny$\pm$0.0016} & 0.1468{\tiny$\pm$0.0024} \\
		         & DASL         & w/o $\mathcal{M}$ & 0.1903{\tiny$\pm$0.0029} & 0.2860{\tiny$\pm$0.0042} & 0.1269{\tiny$\pm$0.0025} & 0.1578{\tiny$\pm$0.0029} \\
		         & PriCDSR-DASL & w/ $\mathcal{M}$  & 0.1791{\tiny$\pm$0.0023} & 0.2755{\tiny$\pm$0.0030} & 0.1174{\tiny$\pm$0.0013} & 0.1484{\tiny$\pm$0.0016} \\
		\hline
		\multirow{3}{*}{Movie$\leftarrow$CD}
		         & DASL-single  & w/o $\mathcal{M}$ & 0.1763{\tiny$\pm$0.0007} & 0.2686{\tiny$\pm$0.0052} & 0.1170{\tiny$\pm$0.0012} & 0.1467{\tiny$\pm$0.0020} \\
		         & DASL         & w/o $\mathcal{M}$ & 0.1942{\tiny$\pm$0.0050} & 0.2858{\tiny$\pm$0.0051} & 0.1286{\tiny$\pm$0.0043} & 0.1580{\tiny$\pm$0.0042} \\
		         & PriCDSR-DASL & w/ $\mathcal{M}$  & 0.1823{\tiny$\pm$0.0054} & 0.2750{\tiny$\pm$0.0039} & 0.1202{\tiny$\pm$0.0027} & 0.1500{\tiny$\pm$0.0020} \\
		\hline
		\multirow{3}{*}{Book$\leftarrow$Movie}
		         & DASL-single  & w/o $\mathcal{M}$ & 0.1759{\tiny$\pm$0.0072} & 0.2681{\tiny$\pm$0.0054} & 0.1171{\tiny$\pm$0.0039} & 0.1466{\tiny$\pm$0.0033} \\
		         & DASL         & w/o $\mathcal{M}$ & 0.1860{\tiny$\pm$0.0018} & 0.2788{\tiny$\pm$0.0010} & 0.1235{\tiny$\pm$0.0017} & 0.1533{\tiny$\pm$0.0016} \\
		         & PriCDSR-DASL & w/ $\mathcal{M}$  & 0.1839{\tiny$\pm$0.0035} & 0.2821{\tiny$\pm$0.0070} & 0.1211{\tiny$\pm$0.0033} & 0.1527{\tiny$\pm$0.0044} \\
		\hline
		\multirow{3}{*}{Book$\leftarrow$CD}
		         & DASL-single  & w/o $\mathcal{M}$ & 0.1759{\tiny$\pm$0.0072} & 0.2685{\tiny$\pm$0.0091} & 0.1160{\tiny$\pm$0.0039} & 0.1459{\tiny$\pm$0.0042} \\
		         & DASL         & w/o $\mathcal{M}$ & 0.1830{\tiny$\pm$0.0035} & 0.2803{\tiny$\pm$0.0002} & 0.1216{\tiny$\pm$0.0025} & 0.1529{\tiny$\pm$0.0021} \\
		         & PriCDSR-DASL & w/ $\mathcal{M}$  & 0.1832{\tiny$\pm$0.0045} & 0.2802{\tiny$\pm$0.0034} & 0.1210{\tiny$\pm$0.0010} & 0.1520{\tiny$\pm$0.0025} \\
		\hline
		\multirow{3}{*}{CD$\leftarrow$Movie}
		         & DASL-single  & w/o $\mathcal{M}$ & 0.1833{\tiny$\pm$0.0048} & 0.2824{\tiny$\pm$0.0034} & 0.1225{\tiny$\pm$0.0042} & 0.1542{\tiny$\pm$0.0036} \\
		         & DASL         & w/o $\mathcal{M}$ & 0.1991{\tiny$\pm$0.0123} & 0.2990{\tiny$\pm$0.0127} & 0.1322{\tiny$\pm$0.0106} & 0.1642{\tiny$\pm$0.0106} \\
		         & PriCDSR-DASL & w/ $\mathcal{M}$  & 0.2025{\tiny$\pm$0.0150} & 0.3084{\tiny$\pm$0.0098} & 0.1348{\tiny$\pm$0.0123} & 0.1688{\tiny$\pm$0.0105} \\
		\hline
		\multirow{3}{*}{CD$\leftarrow$Book}
		         & DASL-single  & w/o $\mathcal{M}$ & 0.1615{\tiny$\pm$0.0044} & 0.2488{\tiny$\pm$0.0004} & 0.1066{\tiny$\pm$0.0041} & 0.1346{\tiny$\pm$0.0028} \\
		         & DASL         & w/o $\mathcal{M}$ & 0.1875{\tiny$\pm$0.0051} & 0.2851{\tiny$\pm$0.0048} & 0.1239{\tiny$\pm$0.0049} & 0.1552{\tiny$\pm$0.0048} \\
		         & PriCDSR-DASL & w/ $\mathcal{M}$  & 0.2052{\tiny$\pm$0.0061} & 0.3054{\tiny$\pm$0.0027} & 0.1381{\tiny$\pm$0.0057} & 0.1703{\tiny$\pm$0.0048} \\
		\hline
	\end{tabular}
\end{table*}

\subsection{Parameter Analysis (RQ3)}
We investigate the impact of the privacy budget $\epsilon$ on model performance, focusing on two data, i.e., Movie$\leftarrow$Book and CD$\leftarrow$Book.
We make this choice because the former results in lower performance of PriCDSR-DASL than DASL, and the later results in higher performance of PriCDSR-DASL than DASL.
We evaluate our PriCDSR-DASL with $\epsilon \in \{1,2,5,10,20,50\}$.
It seems that the values of the privacy budget are larger than that in some existing works \cite{dp-ldp, FCMF}.
However, this is acceptable since these works require noise to be added at each iteration.
As the number of iterations increases, the required privacy budget will increase.
Unlike them, our PriCDSR only needs to add noise once.
Although the privacy budget for this operation is relatively large, it will not increase again.
Moreover, a related work, PriCDR~\cite{PriCDR}, also employs a relatively large privacy budget.
We report the experimental results in Table \ref{tab:result2}, from which, we can have the following observations.
(1)
On the data Movie$\leftarrow$Book, the recommendation performance generally increases as the noise decreases (i.e., the value of $\epsilon$ increases).
The recommendation performance reaches the highest when $\epsilon=20$.
We think that the introduction of noise leads to this counterintuitive result.
The added noise can be regarded as generated pseudo samples, which may lead to improved performance of the recommendation model~\cite{psudo-item-1,psudo-item-2}.
(2)
On the data CD$\leftarrow$Book, the recommendation performance generally decreases as the value of $\epsilon$ increases.
The recommendation performance reaches the highest when $\epsilon=10$, which may also be due to the generated pseudo samples.

\begin{table*}[htbp]
	\caption{Recommendation performance of our PriCDSR-DASL with different privacy budget $\epsilon$.}
	\label{tab:result2}
	\centering
	\begin{tabular}{llcccc}
		\hline
		Datasets                               & Algorithms                  & HR@5                     & HR@10                    & NDCG@5                   & NDCG@10                  \\
		\hline
		\multirow{7}{*}{Movie$\leftarrow$Book} & PriCDSR-DASL($\epsilon=1$)  & 0.1713{\tiny$\pm$0.0014} & 0.2659{\tiny$\pm$0.0055} & 0.1123{\tiny$\pm$0.0007} & 0.1426{\tiny$\pm$0.0009} \\
		                                       & PriCDSR-DASL($\epsilon=2$)  & 0.1761{\tiny$\pm$0.0023} & 0.2751{\tiny$\pm$0.0054} & 0.1160{\tiny$\pm$0.0028} & 0.1478{\tiny$\pm$0.0034} \\
		                                       & PriCDSR-DASL($\epsilon=5$)  & 0.1750{\tiny$\pm$0.0020} & 0.2680{\tiny$\pm$0.0016} & 0.1155{\tiny$\pm$0.0027} & 0.1453{\tiny$\pm$0.0025} \\
		                                       & PriCDSR-DASL($\epsilon=10$) & 0.1791{\tiny$\pm$0.0023} & 0.2755{\tiny$\pm$0.0030} & 0.1174{\tiny$\pm$0.0013} & 0.1484{\tiny$\pm$0.0016} \\
		                                       & PriCDSR-DASL($\epsilon=20$) & 0.1920{\tiny$\pm$0.0023} & 0.2886{\tiny$\pm$0.0031} & 0.1277{\tiny$\pm$0.0018} & 0.1587{\tiny$\pm$0.0022} \\
		                                       & PriCDSR-DASL($\epsilon=50$) & 0.1856{\tiny$\pm$0.0086} & 0.2815{\tiny$\pm$0.0107} & 0.1239{\tiny$\pm$0.0048} & 0.1548{\tiny$\pm$0.0055} \\
		                                       & DASL                        & 0.1903{\tiny$\pm$0.0029} & 0.2860{\tiny$\pm$0.0042} & 0.1269{\tiny$\pm$0.0025} & 0.1578{\tiny$\pm$0.0029} \\
		\hline
		\multirow{7}{*}{CD$\leftarrow$Book}    & PriCDSR-DASL($\epsilon=1$)  & 0.1995{\tiny$\pm$0.0017} & 0.3046{\tiny$\pm$0.0016} & 0.1321{\tiny$\pm$0.0030} & 0.1659{\tiny$\pm$0.0025} \\
		                                       & PriCDSR-DASL($\epsilon=2$)  & 0.1999{\tiny$\pm$0.0003} & 0.3022{\tiny$\pm$0.0044} & 0.1330{\tiny$\pm$0.0008} & 0.1659{\tiny$\pm$0.0012} \\
		                                       & PriCDSR-DASL($\epsilon=5$)  & 0.1904{\tiny$\pm$0.0074} & 0.2944{\tiny$\pm$0.0124} & 0.1231{\tiny$\pm$0.0066} & 0.1566{\tiny$\pm$0.0082} \\
		                                       & PriCDSR-DASL($\epsilon=10$) & 0.2052{\tiny$\pm$0.0061} & 0.3054{\tiny$\pm$0.0027} & 0.1381{\tiny$\pm$0.0057} & 0.1703{\tiny$\pm$0.0048} \\
		                                       & PriCDSR-DASL($\epsilon=20$) & 0.1906{\tiny$\pm$0.0041} & 0.2913{\tiny$\pm$0.0047} & 0.1278{\tiny$\pm$0.0028} & 0.1602{\tiny$\pm$0.0031} \\
		                                       & PriCDSR-DASL($\epsilon=50$) & 0.1804{\tiny$\pm$0.0048} & 0.2822{\tiny$\pm$0.0044} & 0.1204{\tiny$\pm$0.0057} & 0.1532{\tiny$\pm$0.0049} \\
		                                       & DASL                        & 0.1875{\tiny$\pm$0.0051} & 0.2851{\tiny$\pm$0.0048} & 0.1239{\tiny$\pm$0.0049} & 0.1552{\tiny$\pm$0.0048} \\
		\hline
	\end{tabular}
\end{table*}

\section{Conclusions and Future Work}
\label{sec:conclu}
In this paper, we propose a privacy-preserving cross-domain sequential recommender system called PriCDSR to address the privacy concerns therein.
Our PriCDSR aims to protect the sensitive information in auxiliary domains and enable the legal operation and deployment of cross-domain sequential recommender systems.
To achieve this, we define a new differential privacy (DP), i.e., sequential DP (SDP), on the data of an auxiliary domain, considering both the ID information and the order information.
Then, we describe the steps of our PriCDSR in detail.
In particular, we design a novel random mechanism $\mathcal{M}$ to protect the sensitive information of the auxiliary domain.
We theoretically prove that the random mechanism $\mathcal{M}$ satisfies SDP.
This mechanism operates on the data from the auxiliary domains, producing output that can be directly utilized as input for the CDSR method in the target domain.
Therefore, we do not need to make any changes to the base model in the target domain to protect the users' privacy.
In other words, our PriCDSR method is non-invasive.
To the best of our knowledge, we are the first to investigate privacy issues in cross-domain sequential recommender systems.
Moreover, we empirically evaluate how much our PriCDSR sacrifices the recommendation performance.
The experimental results show that the sacrificed recommendation accuracy is smaller than the improvement brought by the introduction of the auxiliary data.

Although our PriCDSR is designed for cross-domain sequential recommender systems, its core component random mechanism $\mathcal{M}$ is not limited to those settings.
It can perhaps be applied in more ways.
For example, users perturb their own data locally rather than on the organization's server to protect their privacy.
For another example, the random mechanism $\mathcal{M}$ may be used for privacy protection of natural language sentences, as long as words are regarded as item IDs in the users' interaction sequences.
We will conduct more research in the future.

\section*{Acknowledgment}

We thank the support of National Natural Science Foundation of China No. 62172283, No. 62272315 and No. 61836005.
We thank Mr. Zitao Xu for his assistance and helpful discussions.

\bibliographystyle{IEEEtran}
\bibliography{IEEEabrv,PriCDSR-ref}

\end{document}